\documentclass[11pt]{article}
\usepackage{fullpage}

\usepackage{hyperref}       
\usepackage{url}            
\usepackage{booktabs}       
\usepackage{amsfonts}       
\usepackage{nicefrac}       
\usepackage{microtype}      
\usepackage{xcolor}         

\usepackage{graphicx} 
\usepackage{amsfonts}
\usepackage{amsthm}
\usepackage{amsmath}

\newtheorem{theorem}{Theorem}

\newtheorem{lemma}[theorem]{Lemma}
\newtheorem{corollary}[theorem]{Corollary}

\newcommand{\X}{\mathcal{X}}
\newcommand{\E}{\mathbb{E}}
\newcommand{\N}{\mathbb{N}}
\newcommand{\R}{\mathbb{R}}

\allowdisplaybreaks
\title{\bf The Robotaxi Placement Problem:\\Minimizing Expected ETA for Stochastic Demand}

\author{Ioannis Caragiannis\\Aarhus University \&\\Google Research\\\url{iannis@cs.au.dk} \and Kostas Kollias\\Google Research\\\url{kostaskollias@google.com} \and Mohammad Roghani\\Google Research\\\url{roghani@google.com} \and Aaron Schild\\Google Research\\\url{aschild@google.com} \and Ali Kemal Sinop\\Google Research\\\url{asinop@google.com}}

\date{}

\allowdisplaybreaks

\begin{document}

\maketitle

\begin{abstract}
Autonomous ride-hailing platforms must strategically position idle robotaxis to minimize the wait times of prospective riders. We formalize this as the \emph{robotaxi placement problem} ($k$-RP). Given a finite metric space and a demand distribution over its points, the goal is to position $k$ robotaxis to minimize the expected total distance in a perfect matching between the robotaxis and $k$ random riders. We present several theoretical results for this stochastic optimization problem. First, we observe that sampling robotaxi locations independently according to the demand distribution yields a randomized $2$-approximation algorithm. Second, we present an explicit inapproximability bound via a novel gap-preserving reduction from the maximum coverage problem. Furthermore, while it is not even clear whether the exact expected cost of a placement can be computed efficiently on general metrics, we design an exact polynomial-time dynamic programming algorithm for $k$-RP in tree metrics by decoupling the stochastic matching dependencies. Finally, empirical evaluations on real-world ride-hailing data reveal that a variance-reduced random placement strategy is highly effective in practice, yielding expected wait times that are very close to those obtained by computationally heavy exact algorithms for the uniform capacitated $k$-median problem.
\end{abstract}

\section{Introduction}

The rise of autonomous ride-hailing services, commonly known as robotaxis, promises to fundamentally transform urban mobility. For both platforms and riders, minimizing the \emph{estimated time of arrival} (ETA) is arguably the most critical metric for success. From the platform's perspective, low ETAs translate to higher rider satisfaction, increased retention rates, and vastly improved fleet utilization. From the rider's perspective, wait time is a primary factor in choosing a ride-hailing service over personal vehicle ownership or public transit. Behind the scenes, the mathematical engine driving this efficiency is the platform's matching algorithm, which continuously pairs available vehicles with incoming rider requests. The quality of these matchings directly dictates the operational viability of the entire business model~\cite{alonsomora2017ondemand}.

Major industry players such as Waymo, Tesla, and Apollo Go face a continuous operational challenge: when vehicles are idle, where should they be proactively positioned (or ``parked'') to best anticipate future demand? Unlike traditional ride-hailing platforms like Uber and Lyft (as well as conventional taxi services), where human drivers independently decide where to cruise or park, robotaxi fleets are fully centrally controlled. This centralized architecture provides platforms with the ability to explicitly optimize the spatial distribution of their idle fleet. However, this ability introduces a fundamental algorithmic challenge: anticipating demand and prepositioning vehicles such that the eventual matchings yield the lowest possible wait times.

Motivated by this, we introduce the \emph{robotaxi placement problem} ($k$-RP). We are given a finite metric space representing the city's road network, a fleet of $k$ robotaxis, and a probability distribution over the metric space indicating where prospective riders are likely to appear. The goal is to determine the optimal locations for the $k$ robotaxis. Crucially, the objective is not just to place vehicles close to high-probability areas, but to minimize the \emph{expected total distance} of the minimum-weight perfect matching between the $k$ robotaxis and $k$ random riders drawn independently from the demand distribution. 

While $k$-RP shares conceptual similarities with classical \emph{facility location} problems, particularly the $k$-median problem, it fundamentally departs from them in two key dimensions. The first is the rigid matching requirement. In the standard uncapacitated $k$-median problem, a single facility can serve an unlimited number of local demands. In contrast, $k$-RP imposes a strict vehicle capacity of exactly one. If multiple riders appear in the same geographical area, only one can be served by the nearest robotaxi; the rest must be routed to vehicles further away, causing cascading ETA delays. The second dimension is its \emph{stochastic optimization} nature. Rather than optimizing over a fixed, known set of demand points, the placement decisions in $k$-RP must be made \emph{a priori}. The true demand is realized only after the robotaxis are positioned, meaning the objective requires minimizing the expected cost of a matching evaluated over exponentially many potential demand realizations. Together, this combination of strict capacity constraints and stochastic demand makes the problem highly non-trivial to analyze, positioning our work as an important theoretical bridge between classical facility location and modern stochastic fleet management.

\subsection{Our contributions and techniques}
We present several theoretical results for the robotaxi placement problem. Our first result is a randomized $2$-approximation algorithm for $k$-RP, which follows from the observation that the perfect matching distance between two $k$-multisets satisfies the triangle inequality. By sampling $k$ points independently from the demand distribution, the triangle inequality combined with the linearity of expectation immediately yields a simple but elegant $2$-approximation bound. 

Our second contribution is establishing the computational hardness of the problem. We show that $k$-RP is APX-hard, proving an explicit inapproximability bound of $1.0234$. We employ a novel gap-preserving reduction from the maximum coverage problem~\cite{feige1998threshold}. The crux of the analysis involves tightly bounding the expected cost of the matchings in the ``partial covering'' case, which requires overcoming the technical challenge of analyzing the complex dependencies that arise when matching independent random riders to the (capacity-constrained) robotaxis.

Furthermore, we show that when the underlying space is a tree metric, the problem becomes tractable. This is particularly notable because simply \emph{evaluating} the exact expected cost of a fixed placement of $k$ robotaxis on a general metric does not seem to be easy (we suspect it is \#P-hard). For tree metrics, however, we design an exact $O(n \cdot k^3)$ time dynamic programming algorithm. The key insight is that an optimal matching has a ``no-crossing'' property, which allows us to completely decouple the stochastic dependencies of the matchings. We can instead calculate the contribution of each tree edge to the expected ETA independently, allowing a dynamic program to compute the exact optimum in a bottom-up fashion.

Finally, we complement our theoretical findings by conducting an empirical evaluation using real-world ride-hailing data from New York City. We assess the practical performance of our random placement algorithm, a variance-reduced variant (VRRP) of it, and a computationally heavy exact solver for the uniform capacitated $k$-median problem (UC$k$M). Our experiments reveal that while UC$k$M yields the best expected wait times, its computational overhead may be prohibitive for large instances. Strikingly, the VRRP strategy bridges this gap; it significantly outperforms basic random placement and delivers wait times very close to the exact baseline, demonstrating that variance-reduced sampling yields highly effective and scalable placements in practice.

\subsection{Related work}
Our robotaxi placement problem bridges foundational concepts in theoretical stochastic optimization and facility location with modern operational challenges in autonomous ride-hailing.

In the theoretical computer science literature, $k$-RP can be viewed through the lens of \emph{a priori optimization}, a framework pioneered by Bertsimas~\cite{bertsimas1988probabilistic, bertsimas1990priori} for problems like the probabilistic traveling salesperson. In an a priori model, a baseline structural solution is constructed before demand is realized, and a predefined recourse strategy is applied once the specific random variables are instantiated. Similarly, $k$-RP acts as a two-stage stochastic optimization problem~\cite{shmoys2004approximation, ravi2004multicommodity}, where the first stage commits to $k$ robotaxi locations, and the second stage matching provides the recourse. However, existing two-stage stochastic facility location frameworks typically allow open facilities to serve multiple demands up to a loose capacity, or they completely lack strict capacity constraints. In contrast, $k$-RP enforces a rigid one-to-one matching requirement upon the realized demand.

This strict capacity constraint deeply connects our formulation to the capacitated $k$-median problem. While the uncapacitated version of the problem has been extensively studied---culminating in the recent breakthrough that yields a $(2+\epsilon)$-approximation algorithm by Cohen-Addad et al.~\cite{cohenaddad2022breakthrough}---enforcing capacities makes the problem historically much more challenging to approximate~\cite{chuzhoy2005capacitated, byrka2015approximation}. A common approach in the literature to bypass this difficulty is to design bicriteria approximation algorithms that violate the capacities by a small factor~\cite{charikar2002constant, KPR00}. However, such relaxations are fundamentally incompatible with $k$-RP, which requires a rigid minimum-weight perfect matching without any allowable capacity violations. Even for the strict uniform capacitated $k$-median problem, true constant-factor approximations without capacity violations have been remarkably difficult to obtain, with major progress achieved by Li~\cite{Li17}, alongside recent advances from the perspective of parameterized complexity by Cohen-Addad and Li~\cite{cohenaddadL2019parameterized}. 

Our techniques are also inspired by more distant problem areas. For example, 
evaluating the expected cost of such random matchings traces its theoretical roots back to foundational work on the optimal matching of random points in metric spaces, famously captured by the AKT theorem~\cite{ajtai1984optimal}, as well as bounds for matching random demands to fixed deterministic locations~\cite{leighton1989tight}. We also note that the $k$-multiset metric we define to achieve our randomized approximation is conceptually related to the Wasserstein (or earth mover's) distance~\cite{rubner2000earth}, applied here in a discrete, combinatorial setting. Finally, regarding computational complexity, our hardness of approximation results build upon the classical bounds for the maximum coverage problem, famously established by Feige~\cite{feige1998threshold}.

Beyond the theoretical community, our work is heavily motivated by the operations research, transportation, and AI literature on ride-sharing and dynamic fleet repositioning. The operational necessity of repositioning empty vehicles to anticipate demand has been widely studied~\cite{pavone2012robotic}. Approaches range from continuous queueing networks and fluid approximations~\cite{braverman2019empty} to receding horizon and model predictive control~\cite{miao2016taxi, iglesias2018data}, as well as dynamic trip-vehicle assignment systems~\cite{alonsomora2017ondemand}. In recent years, data-driven approaches---particularly multi-agent deep reinforcement learning---have become popular for dispatching and repositioning in large-scale networks~\cite{lin2018efficient, holler2019deep}. Because evaluating exact spatial placements over time suffers from the curse of dimensionality, this applied literature adeptly handles the temporal and dynamic complexities using heuristics, fluid limits, and learning algorithms. Our work complements this literature by abstracting the core spatial challenge. By isolating the static snapshot of the stochastic placement problem, we provide the foundational, rigorous theoretical approximation guarantees and complexity bounds that underlie these systems.

\section{Preliminaries}
The robotaxi placement problem we consider is a generalization of the median problem, which is defined as follows. We are given a metric space $(\X, d)$ where $\X$ is a set of $n$ points and $d: \X \times \X \to \mathbb{R}_{\geq 0}$ is a distance function and a probability distribution $P$ over the points of $\X$, i.e., such that $\sum_{x\in \X}{P(x)} = 1$. The median problem is to select a point $x$ from $X$ that minimizes the quantity $\sum_{y\in \X}{P(y)\cdot d(x,y)}$.

In the $k$-RP problem (standing for robotaxi placement), we are again given a metric space $(\X, d)$, where $\X$ is a set of $n$ points and $d: \X \times \X \to \mathbb{R}_{\geq 0}$ is a distance function, and a probability distribution $P$ over $\X$ such that $\sum_{x\in \X}{P(x)} = 1$. In the robotaxi placement jargon, the metric space indicates the map of locations in which robotaxis and riders can lie, and the probability $P(x)$ denotes the probability that a rider can show up at point $x$.

Denote by $\X_k$ the set of all $k$-sized multisets (or, simply, $k$-multisets) of points from $\X$. For two $k$-multisets of points $U$ and $V$ from $\X_k$, we define the distance $d_k(U,V)$ between $U$ and $V$ to be the cost of the minimum-cost perfect matching in the complete bipartite graph defined by the two node sides $U$ and $V$ so that the cost of edge $(u,v)$ is equal to $d(u,v)$. We can view the $k$-multiset $U$ as the locations of $k$ robotaxis and the $k$-multiset $V$ as the locations of $k$ riders. Then the distance $d_k(U,V)$ is the total ETA experienced by the riders in their best possible matching with the available robotaxis. Denote by $P_k$ the probability distribution that returns a $k$-multiset of points from $\X_k$ by selecting $k$ points from $\X$ i.i.d.~according to $P$. 

The objective of the $k$-RP problem is to find a $k$-multiset of points $S\in \X_k$ so that the quantity $\sum_{X\in \X_k}{P_k(X)\cdot d_k(S,X)}$ (or, equivalently, $\E_{X\sim P_k}[d_k(S,X)]$) is minimized. We usually refer to this quantity as the cost of the $k$-RP solution $S$. Hence, $k$-RP for the metric space $(\X,d)$ and the probability distribution $P$ is equivalent to the median problem for the metric space $(\X_k,d_k)$ and the probability distribution $P_k$. In words, we want to find the best placement of $k$ robotaxis so that the expected total ETA for serving $k$ random riders is minimized.

\section{Approximating $k$-RP}\label{sec:apx}
We begin with an important observation.
\begin{lemma}\label{lem:metric}
    The pair $(\X_k,d_k)$ forms a metric space.\footnote{Alternatively, we can think of set of $k$ robotaxis as a probability distribution over the points of the metric space consisting of $k$ probability masses of $1/k$ distributed among the points. Then, the distance $d_k$ between two sets of points is essentially the Wasserstein distance between the corresponding distributions, multiplied by $k$.}
\end{lemma}

\begin{proof}
We just prove that the distance $d_k$ satisfies the triangle inequality; non-negativity, indentity of indiscernibles, and symmetry follow trivially since $(\X,d)$ is a metric space. For a perfect matching $M$ between the $k$-multisets of points $U$ and $V$ from $\X_k$, we denote by $M(u)$ the point in $V$ to which point $u$ is matched. For every pair of $k$-multisets $U$ and $V$ from $\X_k$, denote by $M_{UV}$ a perfect matching between the points in $U$ and $V$ that defines the distance $d_k(U,V)$, i.e., $M_{UV}$ minimizes the quantity $\sum_{u\in U}{d(u,M(u))}$ over all perfect matchings $M$ between the points of $U$ and $V$. 

Now consider three $k$-multisets $U$, $V$, and $W$ of $\X_k$. Let $M$ be a perfect matching that assigns the point $u\in U$ to the point $M_{VW}(M_{UV}(u))$. Notice that each point $u \in U$ is matched to a distinct point of $V$ under $M_{UV}$ and each point $v\in V$ is matched to a distinct point of $W$ under $M_{VW}$. Hence, $M$ is indeed a perfect matching between the $k$-multisets $U$ and $W$. We have
\begin{align*}
d_k(U,W) &= \sum_{u\in U}{d(u,M_{UW}(u))} \leq \sum_{u\in U}{d(u,M(u))} \\
&= {\sum_{u\in U}{d(u,M_{VW}(M_{UV}(u)))}}\\
&\leq \sum_{u\in U}{\left(d(u,M_{UV}(u))+d(M_{UV}(u),M_{VW}(M_{UV}(u)))\right)}\\
&=\sum_{u\in U}{d(u,M_{UV}(u))}+\sum_{v\in V}{d(v,M_{VW}(v))}\\
&=d_k(U,V)+d_k(V,W),
\end{align*}
implying that $d_k$ satisfies the triangle inequality. The first and fourth equalities follow by the definitions of distance $d_k$. The second one follows by the definition of matching $M$ and the third one follows since for each point $u\in U$, $M_{UV}(u)$ is a distinct point of $V$. The first inequality follows by the definition of the perfect matching $M_{UW}$. The second one follows by the triangle inequality of the distance $d$ for points $u$, $M_{UV}(u)$, and $M_{VW}(M_{UV}(u))$.
\end{proof}

We now consider the {\em random placement} algorithm that selects a random $k$-multiset from $\X_k$ according to $P_k$. Equivalently, we form a random $k$-multiset by selecting $k$ i.i.d.~points from $\X$ according to $P$. 

\begin{theorem}\label{thm:rp-main}
Consider a $k$-RP instance and let $O$ be the optimal $k$-multiset and $S$ the random $k$-multiset returned by the random placement algorithm. It holds that $\E_{S\sim P_k}\E_{X\sim P_k}\left[d_k(S,X)\right] \leq 2\cdot \E_{X\sim P_k}\left[d_k(O,X)\right]$.
\end{theorem}

\begin{proof}
We have that 
\begin{align*}
    \E_{S\sim P_k}\E_{X\sim P_k}\left[d_k(S,X)\right] &\leq \E_{S\sim P_k}\E_{X\sim P_k}\left[d_k(S,O)+d_k(O,X)\right]\\
    &=\E_{S\sim P_k}\left[d_k(O,S)\right]+\E_{X\sim P_k}\left[d_k(O,X)\right]\\
    &=2\cdot \E_{X\sim P_k}\left[d_k(O,X)\right],
\end{align*}
as desired. The first derivation follows by the triangle inequality (Lemma~\ref{lem:metric}), the second one is due to linearity of expectation, and the third one is trivial. 
\end{proof}

\begin{corollary}\label{cor:2-apx-det}
    Consider a $k$-RP instance and let $O$ be the optimal $k$-multiset. There exists a $k$-multiset $S$ consisting of points that are drawn with positive probability from $P$ so that $\E_{X\sim P_k}\left[d_k(S,X)\right] \leq 2\cdot \E_{X\sim P_k}\left[d_k(O,X)\right]$.
\end{corollary}

The bound of Corollary~\ref{cor:2-apx-det} is tight, implying that the optimal placement may use locations not included in the support of $P$.

We show that the result in Corollary~\ref{cor:2-apx-det} is tight with the next lemma.
\begin{lemma}
    There exists an instance of $k$-RP with $n$ points and $k=o\left(\sqrt{n}\right)$ such that any $k$-multiset of points that are returned by the distribution $P$ with positive probability satisfies $\E_{X\sim P_k}[d_k(S,X)] \geq (2-o(1))\cdot \E_{X\sim P_k}\left[d_k(O,X)\right]$, where $O$ is the optimal $k$-multiset.
\end{lemma}

\begin{proof}
Consider the instance of $k$-RP with a metric space consisting of a central point $C$ and $n-1$ non-central points at distance $1$ from $C$ and at distance $2$ from each other. The probability distribution $P$ returns each non-central point with probability $\frac{1}{n-1}$. The optimal solution is at least as good as the $k$-multiset in which all points are placed at the central point $C$, yielding an expected cost of $k$. Now, consider any $k$-multiset of non-central points $S$. We have
\begin{align*}
    \E_{X\sim P_k}\left[d_k(S,X)\right] &= \sum_{X\in \X_k}{P_k(X)\cdot d_k(S,X)}\\
    &\geq \sum_{X\in \X_k:S\cap X=\emptyset}{P_k(X)\cdot d_k(S,X)}\\
    &= 2k\cdot \sum_{X\in \X_k:S\cap X=\emptyset}{P_k(X)} \geq  2k\cdot \left(1-\frac{k}{n-1}\right)^k\\
    &\geq \left(2-\frac{2k^2}{n-1}\right) \cdot \E_{X\sim P_k}\left[d_k(O,X)\right]. 
\end{align*}
The first inequality is obvious. The second equality is due to the fact that $d_k(S,X)=2k$ for any two disjoint $k$-multisets $S$ and $X$ of non-central points. The second inequality follows since a random point according to $P$ appears in $S$ with probability at most $\frac{k}{n-1}$ (since $S$ consists of at most $k$ different non-central points). The third inequality follows from the fact that $(1-x)^r\geq 1-x\cdot r$ for $x\in (0,1)$ and $r\geq 1$ and since $\E_{X\sim P_k}\left[d_k(O,X)\right]=k$. Clearly, the RHS of the above derivation is at least $(2-o(1))\cdot \E_{X\sim P_k}\left[d_k(O,X)\right]$ when $k=o\left(\sqrt{n}\right)$.
\end{proof}

\section{Hardness of approximating $k$-RP}\label{sec:apx-hardness}
We devote this section to our APX-hardness result. 

\begin{theorem}\label{thm:apx-hardness}
For every constant $\gamma \in (0,0.0234]$, $k$-RP is NP-hard to approximate within a factor better than $1.0234-\gamma$.
\end{theorem}

\begin{proof}
We use a gap-preserving reduction from the maximum $\ell$-coverage problem \cite{feige1998threshold}. In particular, Feige~\cite{feige1998threshold} proved that given an instance of $\ell$-coverage consisting of a universe of elements $U=[N]$ ($N$ is a multiple of $\ell$), a collection $\mathcal{S}=\{S_1, ..., S_m\}$ of subsets of $U$, each of size $N/\ell$, and an associated constant parameter $\varepsilon>0$, it is NP-hard to distinguish between the following two cases:
\begin{itemize}
    \item There are $\ell$ sets in $\mathcal{S}$ that cover all elements of $U$ ({\em full covering} instance). 
    \item Any $\ell$ sets from $\mathcal{S}$ cover at most a $1-1/e+\varepsilon$ fraction of the elements of $U$ ({\em partial covering} instance).
\end{itemize}

We define the following reduction that transforms an instance of maximum $\ell$-coverage to instances of $k$-RP. Given an instance $I_{\text{cov}}$ of $\ell$-coverage consisting of a universe $U$ of $N$ elements, a collection $\mathcal{S}$ of $m$ subsets of $U$ of size $N/\ell$, and $\varepsilon>0$, we construct the instance $I_{\text{rp}}$ of $k$-RP that consists of a metric space $(\mathcal{X},d)$ and probability distribution $P$ as follows. Set $k=\ell$. $\mathcal{X}$ has an {\em element point} for each element in $U$ and a {\em set point} for each set of $\mathcal{S}$. Hence, $n=N+m$. With some abuse in notation, we denote a point of $\mathcal{X}$ by its corresponding element or set in $I_{\text{cov}}$. Hence, $\mathcal{X}=U\cup \mathcal{S}$. The distance function $d$ is defined as follows. For every element point $u\in U$ and set point $S\in \mathcal{S}$, we have $d(u,S)=1$ if $u\in S$ and $d(u,S)=2-\varepsilon$ otherwise. For any two distinct element points $u_1,u_2\in U$, we have $d(u_1,u_2)=2$ and for any two distinct set points $S_1,S_2\in \mathcal{S}$, we have $d(S_1,S_2)=1$. Clearly, $d$ satisfies the triangle inequality. The distribution $P$ selects equiprobably among the element points.

We can assume that $I_{\text{cov}}$ is a hard instance obtained by Feige's construction in~\cite{feige1998threshold} with small $\varepsilon$, e.g., $\varepsilon \leq 0.06$.
However, $\varepsilon$ affects the other parameters of $I_{\text{cov}}$ which satisfy\footnote{Recall that the quantity $(1-1/k)^k$ approaches $1/e$ as $k$ approaches infinity. By requiring Equation (\ref{eq:epsilon-defn}), we essentially require that the value of $k$ is a sufficiently large constant. A very similar requirement exists for the number of sets in the maximum coverage hardness of approximation result by Feige~\cite{feige1998threshold} we reduce from.}
\begin{align}\label{eq:epsilon-defn}
    \varepsilon\geq \frac{1}{e}-\left(1-\frac{1}{\ell}\right)^\ell=\frac{1}{e}-\left(1-\frac{1}{k}\right)^k
\end{align}
and have set size $N/\ell$ equal to a (high-degree) polynomial of $1/\varepsilon$. In the following, we will assume that the instance $I_{\text{cov}}$ satisfies the weaker inequality
\begin{align}\label{eq:2k-over-N}
\varepsilon &\geq 2\ell/N =2k/N.
\end{align}

We first prove that the optimal solution of $k$-RP for instance $I_{\text{cov}}$ has a special structure.

\begin{lemma}\label{lem:only-set-points}
    Any optimal solution of $k$-RP for instance $I_{\text{rp}}$ consists of $k$ set points.
\end{lemma}

\begin{proof}
For the sake of contradiction, assume that the optimal solution $K$ of $k$-RP for instance $I_{\text{rp}}$ contains an element point $e$. Let $L$ be a random set of $k$ points, drawn independently according to the probability distribution $P$. Also, let $K'$ be the set of $k$ points obtained from $K$, after replacing the element point $e$ with an arbitrary set point $S$ in $\mathcal{S}$. We will show that the expected cost of the best matching between $K'$ and $L$ is strictly better than the cost of the best matching between $K$ and $L$, contradicting our initial assumption. 

We denote by $M$ the (random) matching between $K$ and $L$ of minimum cost. Also, we denote by $M'$ the (random) matching between $K'$ and $L$ defined as follows: $M'(S)=M(e)$ and $M'(p)=M(p)$ for every point $p\in K\setminus \{e\}$. Clearly,  for every point $p\in K\setminus\{e\}$, the edges $(p,M(p))$ and $(p,M'(p))$ contribute equally to the expected cost of matchings $M$ and $M'$, respectively. Now, observe that point $e$ does not belong to set $L$ with probability $(1-1/N)^k$. Hence, the contribution of edge $(e,M(e))$ to the cost of matching $M$ is $2\cdot (1-1/N)^k> 2\cdot (1-k/N)\geq 2-\varepsilon$, i.e., strictly worse than the contribution of edge $(S,M'(S))$ to the cost of matching $M'$. The inequality follows from Bernoulli's inequality stating $(1-z)^r > 1-z\cdot r$ for $z \in (0,1)$ and $r> 1$ (applied with $z=1/N$ and $r=k$). Overall, the expected cost of matching $M'$ is strictly lower than that of matching $M$, completing the proof.
\end{proof}

Due to Lemma~\ref{lem:only-set-points}, the $k$-RP solutions for instance $I_{\text{rp}}$ we consider in the following consist of $k$ set points. Theorem~\ref{thm:apx-hardness} will follow from the next two consequences of our reduction.
\begin{lemma} \label{lem:full}
If $I_{\text{cov}}$ is a full covering instance, there exists a $k$-RP solution of cost at most $k\cdot (1+1/e)$ for instance $I_{\text{rp}}$.
\end{lemma}

\begin{proof}
    Consider the $k$-RP solution $K$ of instance $I_{\text{rp}}$ consisting of the $k$ set points that correspond to the $k$ sets defining a full covering of $U$ for instance $I_{\text{cov}}$. We use the following algorithm to compute a matching $M$ among the multiset $L$ of $k$ random element points drawn independently according to $P$ and the set of points of $K$. The algorithm considers the points in $L$ in arbitrary order. When considering point $p\in L$, it matches it with a set point of $K$ that is at distance $1$ and has not been matched with another point of $L$ before (if any). Clearly, such a match contributes $1$ to the cost of matching $M$. Let $x$ be the number of points in $L$ matched in this way. After considering all points in $L$, the algorithm matches the $k-x$ element points of $L$ that are still unmatched to the $k-x$ set points of $K$ that are still unmatched; each such match contributes $2-\varepsilon$ to the cost of matching $M$. Overall, the cost of $M$ is $x+(k-x)\cdot (2-\varepsilon)\leq 2k-x$. Since the probability that a random element point is at distance $1$ from a given set point (equivalently, the probability that a random element from $U$ belongs to a given set $S$) is $1/k$, we get that the expected number of set points in $K$ that have at least one element point in $L$ at distance $1$ and, consequently, the expectation of $x$, is $k\left(1-(1-1/k)^k\right)$. Hence, the algorithm will compute a matching of average cost
\begin{align*}
    \E[\text{cost}(M)] &\leq 2k-k\cdot \left(1-\left(1-\frac{1}{k}\right)^k\right) \leq k \cdot \left(1+\frac{1}{e}\right),
\end{align*}
as desired.
\end{proof}
\begin{lemma}\label{lem:partial}
If $I_{\text{cov}}$ is a partial covering instance, any $k$-RP solution for instance $I_{\text{rp}}$ has cost at least $k\cdot (1+e^{-1+2/e}-e^{-1}-3\, (1+1/e)\cdot \varepsilon)$.
\end{lemma}

\begin{proof}
Let $K$ be any $k$-RP solution of $I_{\text{rp}}$ consisting of $k$ set points (allowing repetitions). Denote by $U_0$, $U_1$, and $U_{\geq 2}$ the sets of element points in $U$ (equivalently, the elements of $U$) that are at distance $1$ from zero, one, or at least two set points in $K$ (equivalently, that belong to zero, one, and at least two sets from those corresponding to the set points of $K$), respectively. Now, consider a random multiset $L$ of $k$ element points drawn independently according to probability distribution $P$ and denote by $M$ a (random) matching between the points in multisets $K$ and $L$. Since $I_{\text{cov}}$ is a partial covering instance, we have $|U_0|\geq \left(\frac{1}{e}-\varepsilon\right)\cdot |U|$ and, by linearity of expectation,
\begin{align}\label{eq:constraint-U0}
    \E[|L\cap U_0|] \geq \left(\frac{1}{e}-\varepsilon\right)\cdot k.
\end{align}

Denote by $Y$ the number of set points in $K$ that have all points in $L\cap U_1$ at distance $2-\varepsilon$. For each set point $S\in K$, denote by $U_S$ the set of element points in $U_1$ that are at distance $1$ from $S$. Clearly, the sets $U_S$ for $S\in K$ are disjoint, i.e., $\sum_{S\in K}{|U_S|}=|U_1|$. Now, under the condition that the element point $p$, drawn according to probability distribution $P$, belongs to set $U_1$, the probability that it is at distance $1$ from a given point $S\in K$ is $|U_S|/|U_1|$. Hence, the probability that $S$ is at distance $2-\varepsilon$ from all the points in $L\cap U_1$ is $\left(1-|U_S|/|U_1|\right)^{|L\cap U_1|}$. Hence, the expectation of $Y$ is
\begin{align}\nonumber
\E[Y] &=\E\left[\sum_{S\in K}{\left(1-\frac{|U_S|}{|U_1|}\right)^{|L\cap U_1|}}\right]\\\nonumber
&\geq \E\left[k \cdot \left(1-\frac{1}{k}\sum_{S\in K}{\frac{|U_S|}{|U_1|}}\right)^{|L\cap U_1|}\right]\\\nonumber
&=\E\left[k\cdot \left(1-\frac{1}{k}\right)^{|L\cap U_1|}\right]\\\nonumber
&\geq k\cdot \left(1-\frac{1}{k}\right)^{\E[|L\cap U_1|]} \geq k\cdot \left(\frac{1}{e}-\varepsilon\right)^{\E[|L\cap U_1|]/k}\\\label{eq:exp-Y}
&\geq k\cdot e^{-\E[|L\cap U_1|]/k}\cdot (1-e\cdot \varepsilon).
\end{align}
The first inequality follows by the convexity of function $z^r$ for non-negative $z$ and integer $r\geq 1$, which implies that $\frac{1}{t}\cdot \sum_{i\in [t]}{z_i^r}\geq \left(\frac{1}{t}\cdot \sum_{i\in [t]}{z_i}\right)^r$ (applied for $t=k$, $r=|L\cap U_1|$ and $z_i=1-|U_S|/|U_1|$ for distinct $S\in K$). The second inequality follows by Jensen's inequality and the third one by Equation~(\ref{eq:epsilon-defn}). The fourth inequality follows by the concavity of function $z^r$ for $r\in [0,1]$ which implies that, for $z_2>z_1>0$, the slope of the line connecting points $(0,0)$ and $(z_1,z_1^r)$ is higher than the slope of the line connecting points $(z_1,z_1^r)$ and $(z_2,z_2^r)$, i.e., $\frac{z_1^r}{z_1}\geq \frac{z_2^r-z_1^r}{z_2-z_1}$ and, therefore, $z_1^r\geq z_2^r\cdot \frac{z_1}{z_2}$. The inequality here follows using $z_1=\frac{1}{e}-\varepsilon$, $z_2=\frac{1}{e}$, and $r=\frac{E[|L\cap U_1|]}{k}$.

We remark that, so far, we have silently assumed $|U_1|>0$. If this is not the case, we have $|L\cap U_1|=0$ and the condition ``all points in $L\cap U_1$ are at distance $2-\varepsilon$'' is vacuously satisfied for all point sets in $K$. Hence, $\E[Y]=k$, which again implies Equation (\ref{eq:exp-Y}.

Now, let $T$ be the number of points in $L$ that contribute $2-\varepsilon$ to the cost of matching $M$. Clearly, all points in $L\cap U_0$ have this property. Furthermore, since $Y$ of the set points in $K$ are at distance $2-\varepsilon$ from all points in $L\cap U_1$, we have that at least $|L\cap U_1| - k+Y$ points from $L\cap U_1$ contribute $2-\varepsilon$ to the cost of matching $M$, too. We obtain that
\begin{align}\nonumber
    \E\left[\text{cost}(M)\right]&= \E\left[T\cdot \left(2-\varepsilon\right)+k-T\right]\\\nonumber
    &\geq \E\left[(1-\varepsilon)\cdot \left(|L\cap U_0|+|L\cap U_1|-k+Y\right)+k\right]\\\nonumber
    &\geq (1-\varepsilon)\cdot \E[|L\cap U_0|]+(1-\varepsilon)\cdot \E[|L\cap U_1|]\\\nonumber
    &\quad\quad+\varepsilon\cdot k+(1-\varepsilon)\cdot k\cdot e^{-\E[|L\cap U_1|]/k}\cdot (1-e\cdot \varepsilon)\\\nonumber
    &\geq k \cdot \bigg((1-\varepsilon)\cdot \frac{\E[|L\cap U_0|]}{k}+(1-\varepsilon)\cdot \frac{\E[|L\cap U_1|]}{k}\\\label{eq:bound}
    &\quad\quad+\varepsilon+e^{-\E[|L\cap U_1|]/k}\cdot (1-(e+1)\cdot \varepsilon)\bigg).
\end{align}
The second inequality follows from Equation~(\ref{eq:exp-Y}). The last inequality follows by the fact $(1-\alpha)\cdot (1-\beta)\geq 1-\alpha-\beta$ for $\alpha,\beta\in [0,1]$. We will use this inequality again below.

By the definitions of sets $U_0$, $U_1$, and $U_{\geq 2}$, we have $|U_{\geq 2}|=|U|-|U_0|-|U_1|$ and $|U|\geq |U_1|+2\cdot |U_{\geq 2}|$. Combining them, we get
\begin{align*}
    |U|\geq |U_1|+2\cdot |U_{\geq 2}| &= |U_1|+2\cdot (|U|-|U_0|-|U_1|)\\
    &=2\cdot |U|-2\cdot |U_0|-|U_1|
\end{align*}
and, equivalently, $|U_1|\geq |U|-2\cdot |U_0|$. Thus, by linearity of expectation,
\begin{align}\label{eq:constraint-U1}
    \E[|L\cap U_1|] &\geq k-2\cdot \E[|L\cap U_0|].
\end{align}
In the rest of the proof of Lemma~\ref{lem:partial}, our goal is to lower-bound the parenthesis in the RHS of Equation~(\ref{eq:bound}), subject to the constraints (\ref{eq:constraint-U0}) and (\ref{eq:constraint-U1}). By setting $x=\E[|L\cap U_0|]/k$ and $y=\E[|L\cap U_1|]/k$, we will do so by lower-bounding the objective of the following mathematical program:
\begin{align*}
    \mbox{minimize}\quad\quad & x\cdot (1-\varepsilon)+y\cdot (1-\varepsilon)+\varepsilon +e^{-y}(1-(e+1)\cdot \varepsilon)\\
    \mbox{subject to:}\quad\quad & y \geq \max\{0,1-2x\}\\
    & x\geq \frac{1}{e}-\varepsilon
\end{align*}
Observe that the objective is non-decreasing in $y$. Indeed, since $y\geq 0$, its derivative is $1-\varepsilon-e^{-y}(1-(e+1)\cdot \varepsilon)\geq e\cdot \varepsilon>0$. Hence, if $x\geq 1/2$, the objective function is minimized for $y=0$ to at least $\frac{3}{2}-\left(e+\frac{1}{2}\right)\cdot \varepsilon > 1-\frac{1}{e}+e^{-1+2/e}-3\, \left(1+\frac{1}{e}\right)\cdot \varepsilon$ and the RHS of Equation~(\ref{eq:bound}) is at least $k \cdot \left(1-\frac{1}{e}+e^{-1+2/e}-3\left(1+\frac{1}{e}\right)\cdot \varepsilon\right)$, as desired. 
Otherwise, the objective function is minimized for $y=1-2x$ to $1-x\cdot (1-\varepsilon)+e^{-1+2x}(1-(e+1)\cdot \varepsilon)$. It remains to lower-bound the objective of the following mathematical program:
\begin{align*}
    \mbox{minimize}\quad\quad & 1-x\cdot (1-\varepsilon)+e^{-1+2x}(1-(e+1)\cdot \varepsilon)\\
    \mbox{subject to:}\quad\quad & x\geq \frac{1}{e}-\varepsilon
\end{align*}
The derivative of the new objective with respect to $x$ is clearly increasing. Hence, for $x\geq \frac{1}{e}-\varepsilon$, it evaluates to
\begin{align*}
-1+\varepsilon +2e^{-1+2x}(1-(e+1)\cdot \varepsilon)&\geq -1+\varepsilon+2e^{-1+2/e-2\varepsilon}(1-(e+1)\cdot \varepsilon)\\
&\geq -1+\varepsilon+2e^{-1+2/e}\cdot (1-2\varepsilon)\cdot (1-(e+1)\cdot \varepsilon)\\
&\geq -1+\varepsilon+2e^{-1+2/e}\cdot (1-(e+3)\cdot \varepsilon)> 0.
\end{align*}
The second inequality follows from the property $e^z\geq 1+z$ (applied for $z=-2\varepsilon$). The third one follows from the property we mentioned earlier, and the fourth one by the assumption $\varepsilon\leq 0.06$. Hence, the objective of the second mathematical program is increasing for $x\geq 1/e-\varepsilon$ and is minimized to
\begin{align*}
    &1-\left(\frac{1}{e}-\varepsilon\right)\cdot (1-\varepsilon)+e^{-1+2/e-2\varepsilon}\cdot (1-(e+1)\cdot \varepsilon)\\
    &\geq 1-\frac{1}{e}+\varepsilon+e^{-1+2/e}\cdot (1-2\varepsilon)\cdot (1-(e+1)\cdot \varepsilon) \\
    &\geq 1-\frac{1}{e}+e^{-1+2/e}-\left((e+3)\cdot e^{-1+2/e}-1\right)\varepsilon\\
    &\geq 1-\frac{1}{e}+e^{-1+2/e}-3\left(1+\frac{1}{e}\right)\cdot \varepsilon.
\end{align*}
Hence, the RHS of Equation (\ref{eq:bound}) is again at least $k \cdot \left(1-\frac{1}{e}+e^{-1+2/e}-3\left(1+\frac{1}{e}\right)\cdot \varepsilon\right)$, as desired.
\end{proof}

We are now ready to complete the proof of Theorem~\ref{thm:apx-hardness}. For constant $\gamma\in (0,0.0234]$, we use a hard instance $I_{\text{cov}}$ of $\ell$-coverage with $\varepsilon=\gamma/3<0.06$. Then, by Lemmas~\ref{lem:full} and~\ref{lem:partial}, we conclude that a polynomial-time algorithm which approximates $k$-RP within a factor better than $\frac{1+e^{-1+2/e}-1/e-3\,\left(1+\frac{1}{e}\right)\cdot \varepsilon}{1+1/e}\geq 1.0234-\gamma$ could be used to decide whether $I_{\text{cov}}$ is a full or a partial covering instance of $\ell$-coverage, completing the proof.
\end{proof}

\section{A polynomial-time algorithm for tree metrics}
Our goal in this section is to explore the most general metric space that allows for an exact solution of $k$-RP in polynomial time. We remark that the seemingly simpler task of computing the expected ETA of a given robotaxi placement exactly is a challenging task in general (we suspect it is $\#P$-hard). Indeed, the only tool that seems to be available is sampling, which is inherently inaccurate. Interestingly, when the metric space is a tree metric, the stochastic nature of the problem is no longer an obstacle for computing the cost of a $k$-RP solution exactly. In this case, we can even solve $k$-RP in polynomial time. We explain how to do this below.

In a tree metric, the distance between any pair of nodes is the sum of distances between the endpoints of the edges that belong to the path connecting the two nodes. Hence, an alternative way to measure the cost of a matching between two $k$-sized sets of nodes $X$ and $S$ is to measure the contribution of each edge of the tree to this cost. If an edge $e$ belongs to $n_e(M)$ paths defined by matching $M$ between the nodes of sets $X$ and $S$, the cost of $M$ is $\sum_e{n_e(M)\cdot c(e)}$, where $c(e)$ is the distance between the endpoints of edge $e$.

Consider an edge $e$ and the two subtrees that the removal of edge $e$ defines. We say that the matching $M$ {\em crosses} edge $e$ if there are two nodes $x_1\in X$ and $s_1\in S$ that belong to one of the subtrees, two nodes $x_2\in X$ and $s_2\in S$ that belong to the other subtree, so that $x_1$ is matched to $s_2$ and $x_2$ is matched to $s_1$ in $M$. The next statement follows from the triangle inequality.

\begin{lemma}\label{lemma:no-edge-crossing}
    There exists a minimum-cost perfect matching between the $k$-multisets $X$ and $S$ that does not cross any edge of the tree metric.
\end{lemma}

\begin{proof}
Let $M$ be a minimum-cost matching that crosses an edge $e=(u,v)$, i.e., it contains two nodes $x_1\in X$ and $s_1\in S$ that belong together with node $u$ to one of the subtrees defined by the removal of edge $e$, two nodes $x_2\in X$ and $s_2\in S$ that belong to the other subtree together with node $v$, so that $x_1$ is matched to $s_2$ and $x_2$ is matched to $s_1$ in $M$. We will show that, by replacing the pairs $(x_1,s_2)$ and $(x_2,s_1)$ by the pairs $(x_1,s_1)$ and $(x_2,s_2)$ in $M$, its cost does not increase. Indeed, we have
\begin{align*}
    d(x_1,s_2)+d(x_2,s_1) &= d(x_1,u)+d(u,v)+d(v,s_2)+d(x_2,v)+d(v,u)+d(u,s_1)\\
    &\geq d(x_1,u)+d(v,s_2)+d(x_2,v)+d(u,s_1)\\
    &\geq d(x_1,s_1)+d(x_2,s_2).
\end{align*}
The claim follows by repeating for every edge crossing.
\end{proof}
Now, consider a minimum-cost perfect matching $M$ between the node $k$-multisets $S$ and $X$ that contains no edge crossings. Then, the quantity $n_e(M)$ is equal to the discrepancy of the $k$-multisets $S$ and $X$ in the two subtrees defined by the removal of edge $e$. I.e., denoting by $T_e$ one of the two subtrees defined by the removal of edge $e$ (it does not matter which one as both subtrees define the same quantity), we have $n_e(M)=\big||X\cap T_e|-|S\cap T_e|\big|$. Hence, the objective of the $k$-RP is equivalently to compute a $k$-multiset $S$ from $\X_k$ so that the quantity 
\begin{align}\label{eq:sum-of-edge-contr}
\E_{X\sim P_k}[d_k(S,X)] &= \sum_e{\E_{X\sim P_k}\left[\big||X\cap T_e|-|S\cap T_e|\big|\right]\cdot c(e)}
\end{align}
is minimized.

Notice that, once we have decided that $t$ nodes of set $S$ belong to subtree $T_e$, the expectation in Equation (\ref{eq:sum-of-edge-contr}) is easy to compute since the distribution of the cardinality of set $X\cap T_e$ when the nodes in $X$ are drawn independently according to the probability distribution $P$ is the binomial distribution with $k$ trials and success probability $p_e=\sum_{s \in T_e}{P(s)}$. Define
\begin{align}
    B(t,p) &= \sum_{s=0}^k{|t-s|\cdot \binom{k}{s}\cdot p^s\cdot (1-p)^{k-s}},
\end{align}
and notice that the quantity $B(t,p)$ can be computed in time $O(k)$. Then, Equation (\ref{eq:sum-of-edge-contr}) becomes
\begin{align*}
\E_{X\sim P_k}[d_k(S,X)] &= \sum_e{B(|S\cap T_e|,p_e)\cdot c(e)}.
\end{align*}
Clearly, evaluating the cost of the $k$-RP solution $S$ can be done in time $O(n\cdot k)$.

We will now show how to use dynamic programming to compute an optimal $k$-RP solution. Without loss of generality, we can assume that the tree metric is a binary tree rooted at a node $h$, such that each node $u$ that is not a leaf has two children $\ell(u)$ and $r(u)$.\footnote{Any tree metric can be transformed to an equivalent having this form; all we need to do after rooting the tree at a node is to replace each node with $c>2$ children with a binary subtree that is rooted at the node and has $c$ leaves. The cost of each edge to a leaf node is equal to the cost of the original edge to the node, and the cost of any other edge in the binary subtree is $0$.} We slightly abuse notation to denote by $T_u$ the subtree rooted at node $u$. We also set $p_u=\sum_{s \in T_u}{P(s)}.$

For each node $u$ of the tree and integer $t\in \{0,1, ..., k\}$, we compute the table entry $V(u,t)$ which keeps the minimum contribution of the edges in the subtree $T_u$ to the expected cost of any minimum-cost perfect matching between any $k$-multiset $S$ that contains $t$ elements in the subtree $T_u$ and a multiset of $k$ random nodes, drawn independently from the tree metric $\X$, according to probability distribution $P$. This can be computed in a bottom-up fashion, by setting $V(u,t) = 0$ for every leaf $u$ and integer $t=0, 1, ..., k$, and
\begin{align*}
    V(u,t) &= \min_{\substack{t_\ell,t_r\in \N:\\t_\ell+t_r\leq t}}\left\{V(\ell(u),t_\ell)+B(t_\ell,p_{\ell(u)})\cdot c(u,\ell(u))+V(r(u),t_r)+B(t_r,p_{r(u)})\cdot c(u,r(u))\right\},
\end{align*}
for each non-leaf node $u$ and integer $t=0, 1, ..., k$.
The entry $V(h,k)$ contains the cost of the optimal $k$-RP solution. We can precompute the $n\cdot k$ values $B(\cdot, \cdot)$ in time $O(n\cdot k^2)$ and, then, each of the $n\cdot k$ values of table $V$ in time $O(k^2)$, i.e., time $O(n\cdot k^3)$ in total.

The next statement summarizes the discussion in this section.\footnote{The algorithm can be further optimized to run in time $O(n\cdot k)$. Details are omitted to keep the exposition simple.}

\begin{theorem}
    The dynamic programming algorithm above solves $k$-RP in tree metrics with $n$ node locations in time $O(n\cdot k^3)$.
\end{theorem}

\section{Experiments}\label{sec:exp}
We now describe our experiments. We first present our experimental setup; additional details can be found in Appendix~\ref{sec:exp-setup-details}. We have built a metric space and probability distributions from real data. The metric space has 263 points, corresponding to the centroids of the 263 taxizones of New York City, as defined by the NYC Taxi and Limousine Commission\footnote{\url{https://www.nyc.gov/site/tlc/}} (TLC). The distance between any pair of points is the average commute time between the two points in both directions, as returned by the Open Source Routing Machine\footnote{\url{https://project-osrm.org/}} (OSRM, \cite{luxen2011real}). This required almost 69K queries in their database via the OSRM API. 

In our experiments, we use four probability distributions over these locations, created by processing the 2023 High Volume For-Hire Vehicle Trip Dataset\footnote{\url{https://data.cityofnewyork.us/Transportation/2023-High-Volume-FHV-Trip-Data/u253-aew4/about_data}} provided by TLC. In particular, the four probability distributions A, B, C, and D were defined as the empirical distributions of the pickup locations of the following sets of trips during week 12 of 2023:
\begin{description}
\item[A:] 187,905 trips between 8:00 AM and 9:00 AM during the weekdays (March 20-24, 2023);
\item[B:] 171,623 trips between 5:00 PM and 6:00 PM during the weekdays;
\item[C:] 37,230 trips between 8:00 AM and 9:00 AM during the weekend (March 25-26, 2023);
\item[D:] 177,249 trips between 5:00 PM and 9:00 PM during Saturday, March 25, 2023.
\end{description}

In all our experiments\footnote{All algorithms, data processing pipelines, and experimental evaluations were implemented in Python 3 and executed in the cloud using standard CPU instances within the Google Colaboratory environment.}, we have used $k=1000$. We have implemented the following algorithms: 
\begin{description}
    \item[RP:] This is our random placement algorithm from Section~\ref{sec:apx}.
    \item[VRRP:] This is a variant of RP, which first assigns $\lfloor P(x)\cdot k \rfloor$ robotaxis in each location $x$ of the metric space. The location of each of the remaining $k-\sum_{x}{\lfloor P(x) \cdot k \rfloor}$ robotaxis is selected independently and proportionally to the residual probability $P(x)-\lfloor P(x)\cdot k\rfloor/k$ at each location $x$.
    \item[UC$k$M:] This is an exact algorithm for the uniform capacitated $k$-median problem. The algorithm assigns a robotaxi for each probability mass of $1/k$ at each point.
\end{description}

For computing matchings between sets of $k$ locations, we used extensively the Python Optimal Transport library (POT) and, specifically, the extremely fast \texttt{ot.lp.emd} function~\cite{flamary2021pot}. The implementation of algorithms RP and VRRP is straightforward. 

The objective of the variation of uniform capacitated $k$-median we consider here (and of algorithm UC$k$M) is to compute a $k$-multiset of points in a metric space and a uniform distribution over them that has minimum Wasserstein distance from a given distribution over the points. In the robotaxi placement context, this is equivalent to the artificial problem of locating $k$ robotaxis to serve with minimum total ETA a continuous fractional demand distribution $P$ of total mass $1$, where each robotaxi is assigned a uniform capacity of $1/k$.

\begin{figure*}[t] 
    \centering
    \includegraphics[width=0.24\linewidth]{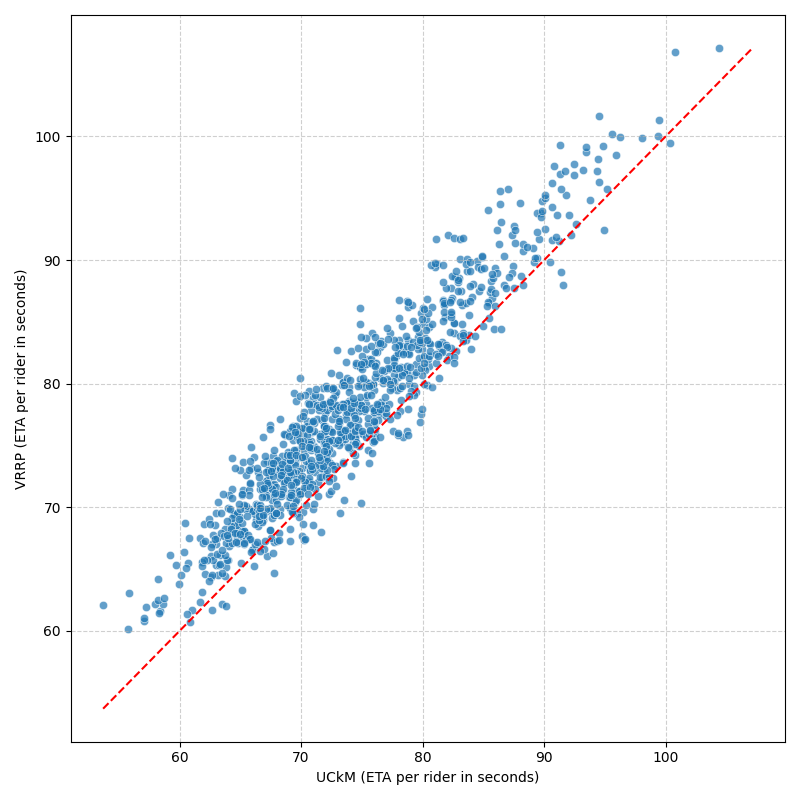}\hfill 
    \includegraphics[width=0.24\linewidth]{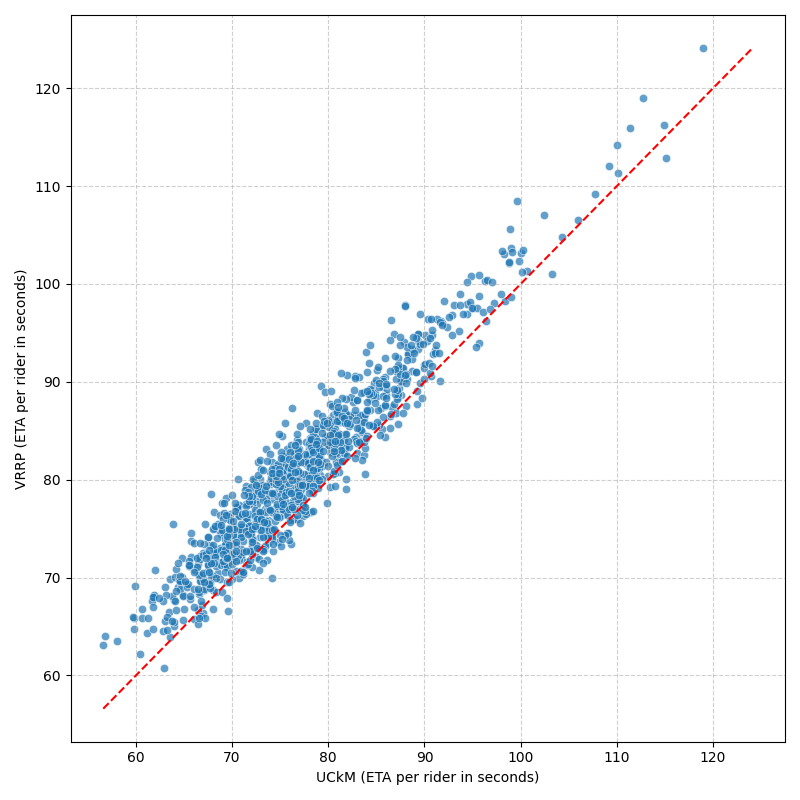}\hfill
    \includegraphics[width=0.24\linewidth]{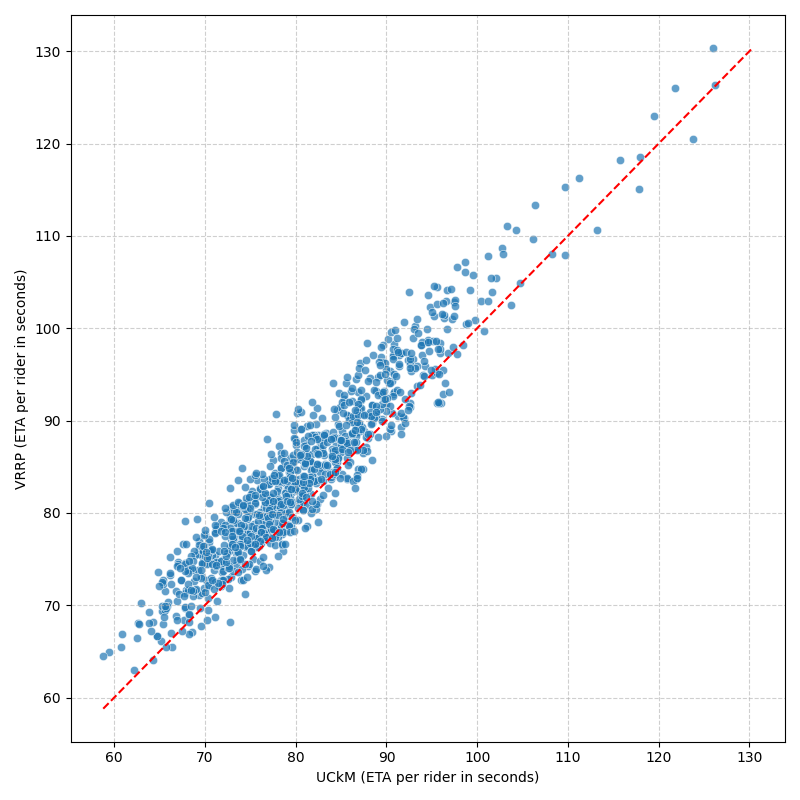}\hfill
    \includegraphics[width=0.24\linewidth]{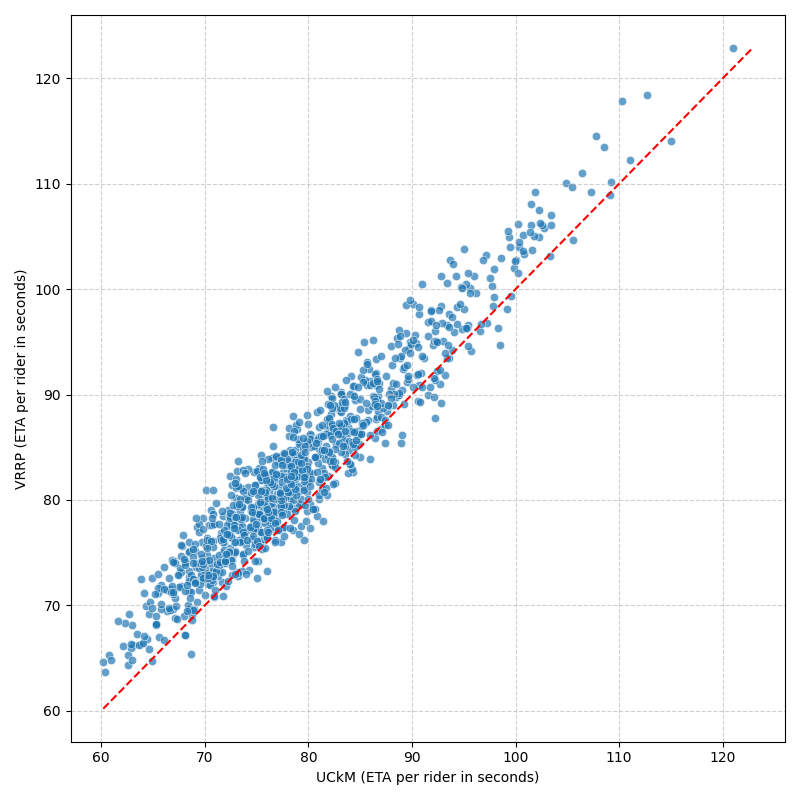}
    \caption{Comparison of algorithm UC$k$M (vertical axis) and the best run of algorithm VRRP (horizontal axis) across 1,000 random rider demand realizations for each of our four distributions A, B, C, and D.}
    \label{fig:four_plots}
\end{figure*}

To compute the UC$k$M solution, we solve a mixed integer linear program with an integer variable $Q(x)$ for each location $x$ indicating the number of robotaxis to be put there (or, equivalently, the multiples of $1/k$ in $Q(x)$) and fractional variables $I(x,y)$ indicating the amount of probability that will be transferred from location $x$ to location $y$. Specifically, the UC$k$M MILP is as follows:
\begin{align*}
    \mbox{minimize} & \sum_{x\in X}\sum_{y\in X}{I(x,y)\cdot d(x,y)}\\
    \mbox{subject to:} & \sum_{y \in X}{I(x,y)} = Q(x)/k, \forall x\in X\\
    & \sum_{x \in X}{I(x,y)} = P(y), \forall y\in X\\
    & \sum_{x\in X}{Q(x)}=k\\
    & I(x,y) \geq 0, \forall x,y\in X\\
    & Q(x) \mbox{ is a non-negative integer}, \forall x\in X
\end{align*}
We solve the MILP using the COIN-OR Branch and Cut (CBC) solver, which is integrated in the PuLP Python library~\cite{forrest_cbc,mitchell2011pulp}. The placement for each of the four instances (with approximately 263 integral and 69,000 fractional variables) was computed by running the CBC solver for (up to) 200 seconds to reach optimality. No solution was obtained with considerably lower time limits.

The results are depicted in Table~\ref{tab:data}. For the algorithms RP and VRRP, we made 100 independent runs of the algorithm and computed the average and minimum per-rider ETA. To estimate the per-rider ETA in each independent run, we consider 1000 sets of 1000 random riders (according to the corresponding probability distribution) each.

\begin{table}[h]
\centering
\begin{tabular}{c|ccc|ccc|c}\hline
 & \multicolumn{3}{c|}{RP} & \multicolumn{3}{c|}{VRRP} & UC$k$M \\
p.d. & avg & 95\% ci & min & avg & 95\% ci & min & \\\hline
A & $111.5$ & $\pm 1.0$ & $102.5$ & $86.3$ & $\pm 0.4$ & $83.5$ & $80.4$\\
B & $109.0$ & $\pm 1.4$ & $96.6$ & $82.9$ & $\pm 0.3$ & $80.9$ & $76.7$ \\
C & $112.7$ & $\pm 1.1$ & $103.3$ & $85.7$ & $\pm 0.4$ & $82.8$ & $79.4$ \\
D & $102.2$ & $\pm 1.1$ & $91.9$ & $79.5$ & $\pm 0.3$ & $77.2$ & $73.7$ \\\hline
\end{tabular}
\caption{Comparison of the three algorithms for the four probability distributions. Each number indicates the per-rider ETA in seconds.}
\label{tab:data}
\end{table}

In all cases, algorithm UC$k$M yields the best results with ETA per rider between $73.7$ (for p.d.~D) and $80.4$ (for p.d.~A). However, this comes at the expense of a huge running time of 200 seconds per instance for the instance parameters we consider. In contrast, algorithms RP and VRRP take only a few milliseconds (the time necessary to sample from the distribution of pickup locations). VRRP has approximately $8\%$ worse ETA for all distributions; such a small gap is rather remarkable given the simplicity of VRRP. The best execution of VRRP (among all independent runs) closes the gap with UC$k$M further. Interestingly, algorithm RP yields significantly worse expected ETAs than both VRRP and UC$k$M. Thus, our Theorem~\ref{thm:rp-main} implies that both VRRP and UC$k$M obtain considerably better than $2$-approximations of the optimal per-rider ETA in practice.\footnote{Theoretical guarantees for algorithms VRRP and UC$k$M are proved in Appendix~\ref{sec:vrrp+uckm}.}. The table includes the 95\% confidence intervals for the 100 independent runs of RP and VRRP. As VRRP is a variance-reduced variant of RP, its 95\% confidence intervals (between $\pm 0.3$ and $\pm 0.4$) computed across the 100 independent runs are consistently lower than those of RP (between $\pm 1.0$ and $\pm 1.4$).

In the four plots of Figure~\ref{fig:four_plots}, we see a comparison between UC$k$M and the best run of VRRP on a per-realization basis. In each plot, each of the 1,000 depicted points corresponds to the execution of the two algorithms on a set of 1,000 random riders (according to the corresponding probability distribution A, B, C, and D). The 95\% confidence intervals for the results of the two algorithms on the 1,000 points are between $\pm 0.2$ and $\pm 0.3$.

\section{Discussion}
In this paper, we formalized and studied the robotaxi placement problem ($k$-RP). Our algorithmic results include a randomized $2$-approximation algorithm for general metrics and an exact polynomial-time dynamic programming algorithm for tree metrics. On the hardness side, we established an explicit APX-hardness bound of $1.0234$. Several compelling theoretical questions remain open. While we have not made a particular attempt to optimize the inapproximability constant, we suspect it can be improved using variations of our current technique. Overall, however, we expect the true approximability of $k$-RP to be closer to our upper bound of $2$. Additionally, exploring whether a polynomial-time approximation scheme exists for Euclidean or planar metrics---reflecting the 2D geometry of real road networks---is a highly promising direction. Moreover, a curious open question is whether  \emph{deterministic} constant-factor approximations exist for $k$-RP. In Appendix \ref{sec:vrrp+uckm}, we show that approximation algorithms for the uniform capacitated $k$-median problem yield similar approximation guarantees for $k$-RP. However, whether constant-approximation algorithms for UC$k$M exist is considered one of the most challenging open problems in the literature. An interesting extension of $k$-RP requires robotaxis to be placed in allowable locations only. In Appendix \ref{sec:rrp}, we provided a $3$-approximation for this restricted variant by extending our positive result for $k$-RP.

In our study, we briefly considered local search heuristics. Even though it was clear that their theoretical performance for $k$-RP is poor, they seemed ideal as alternative algorithms for our experiments. For example, local search could be used to improve the solutions returned by algorithms VRRP and UC$k$M. Unfortunately, the necessity to continuously evaluate neighboring placements by extensive sampling makes local search computationally impractical for the scale of instances we considered. 

Our experiments considered fleets of 1,000 robotaxis, a scale we expect to be typical during the initial years of robotaxi operations in large cities. For instances of this size, formulating the placement via the uniform capacitated $k$-median problem (UC$k$M) and solving it directly is practically effective and yields the best ETA. However, considering that the current number of for-hire vehicles in NYC exceeds 100,000, there is significant room for larger fleets in the future. Furthermore, more refined metrics will be necessary in practice for ETA predictions that are as accurate and, thus, as attractive to users as possible. As the scale of the problem increases, the computational overhead of UC$k$M will inevitably become prohibitive. In these large-scale regimes, our fast, variance-reduced VRRP algorithm will be a highly attractive and scalable alternative.

\bibliographystyle{plain}
\bibliography{refs}

\appendix

\appendix
\section{Experimental setup details}\label{sec:exp-setup-details}
We now present additional details for our experimental setup. We remark that both our metric space and the probability distributions are included together with our code as supplementary material. The next map\footnote{The geographic maps presented in this section have been generated using Python and the Folium library, with underlying map tiles and spatial data provided by OpenStreetMap~\cite{haklay2008openstreetmap} contributors.} (Figure \ref{fig:metric-map}) has a circle at the centroid of each NYC taxizone. These are essentially the points in the metric space we have used in our experiments.

\begin{figure}[htbp] 
    \centering
    \includegraphics[width=0.4\linewidth]{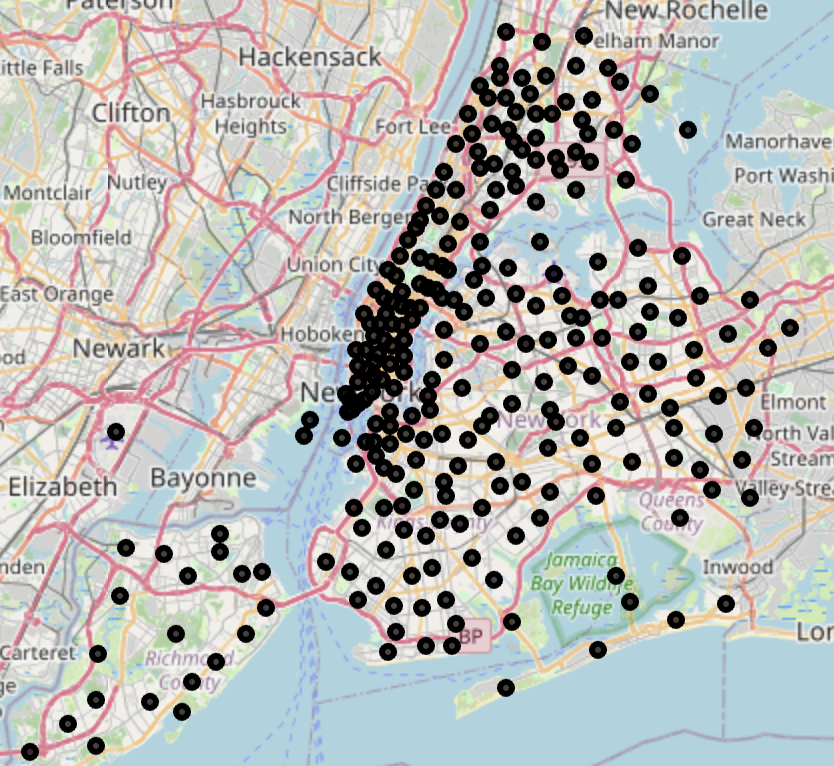}
    \caption{The 263 points of our metric space corresponding to the centroids of the 263 NYC taxizones. }
    \label{fig:metric-map}
\end{figure}

The spatial density of the 263 taxizone centroids is highly non-uniform across the city. The points are densely clustered in Manhattan and the adjacent inner neighborhoods of Brooklyn and Queens, reflecting the smaller geographic footprint of taxizones in the high-traffic urban core. Conversely, the density of points decreases significantly in the outer boroughs, such as Staten Island and eastern Queens, where individual taxizones cover much larger geographical areas.

The four distributions are depicted pictorially in Figure~\ref{fig:pd}. For each taxizone, the probability that a random rider has a pickup location therein is denoted by a circle of radius proportional to the probability and centered at the taxizone centroid. Darker shades and higher radii indicate higher probability. 

\begin{figure}[htbp] 
    \centering
    \includegraphics[width=0.49\linewidth]{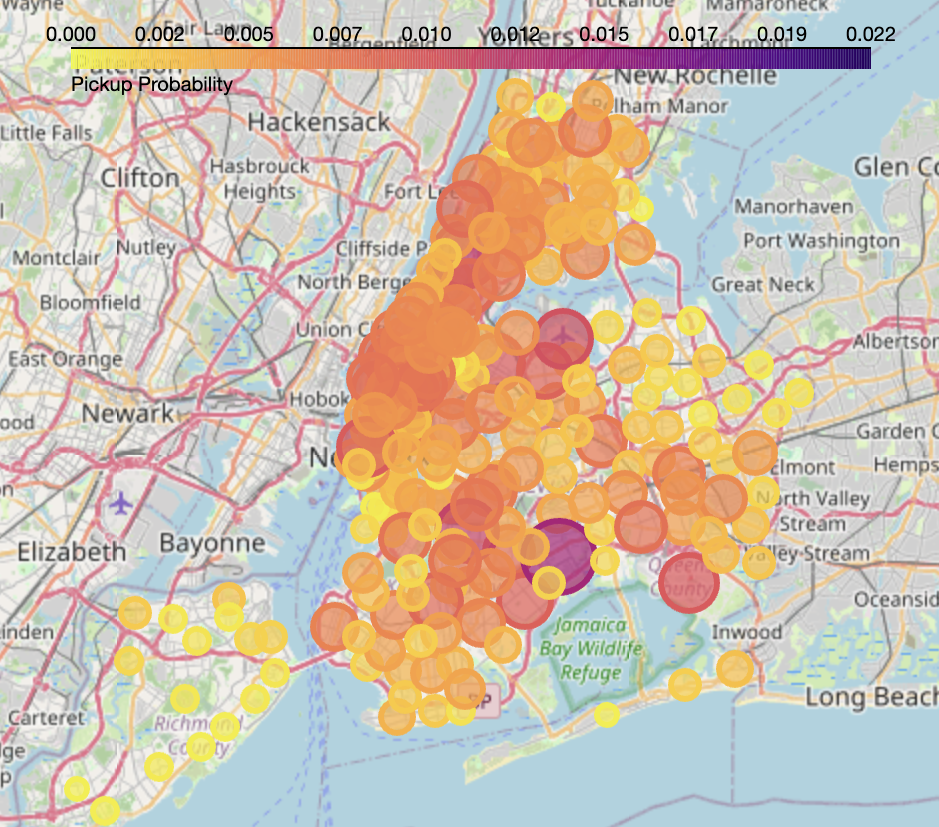}\hfill 
    \includegraphics[width=0.49\linewidth]{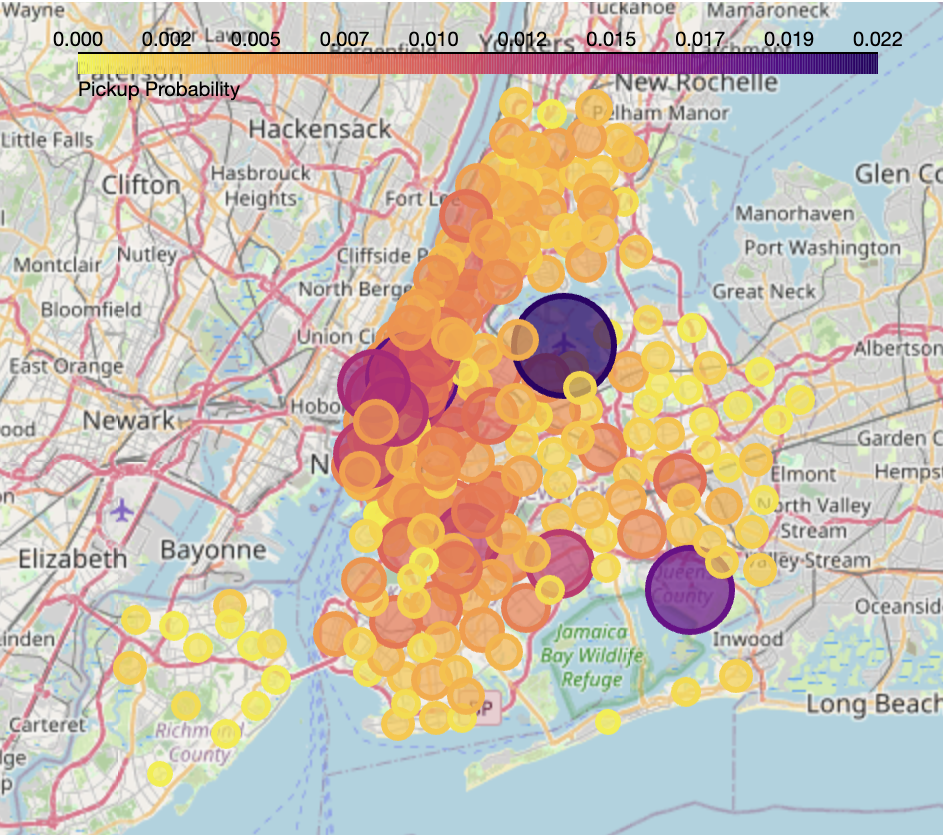}\hfill
    \includegraphics[width=0.49\linewidth]{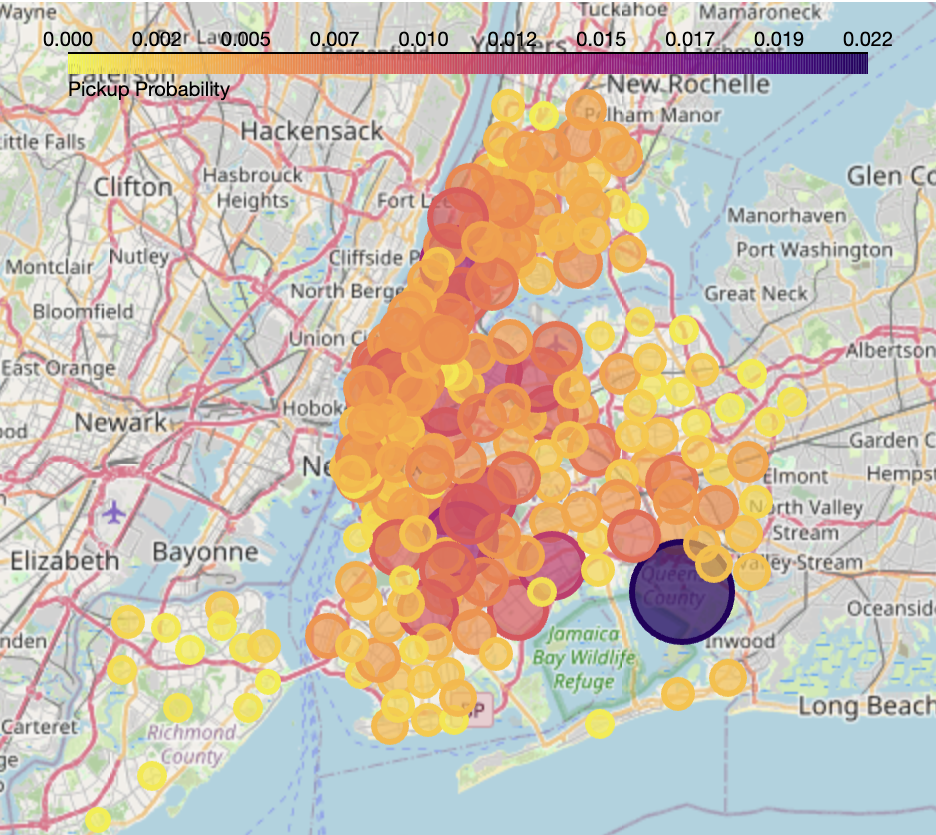}\hfill
    \includegraphics[width=0.49\linewidth]{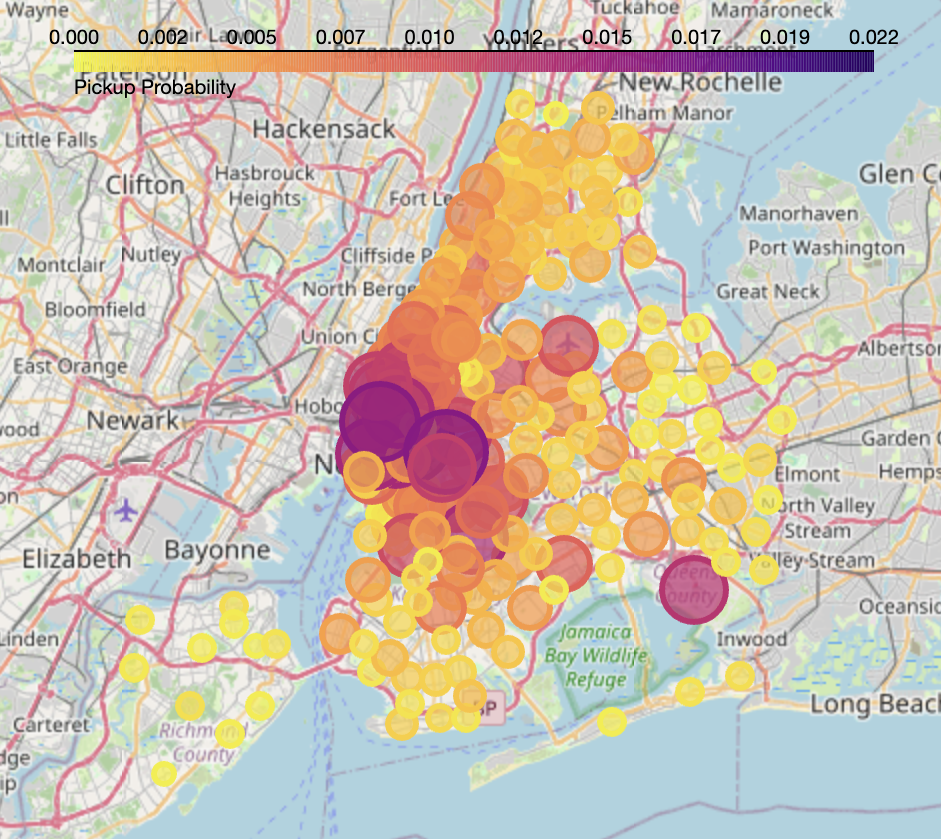}
    \caption{Visual representation of the four probability distributions over the riders' pickup locations used in our experiments: p.d.~A with pickups during Weekdays 8:00-9:00 AM of week 12 of 2023 (upper-left), p.d.~B, Weekdays 5:00-6:00 PM (upper-right), p.d.~C, Weekend 8:00-9:00 AM (lower-left), p.d.~D, Saturday 5:00-9:00 PM (lower-right).}
    \label{fig:pd}
\end{figure}

P.d.~A (upper-left) is based on pickups during the morning commuter rush of the weekdays (8:00 AM – 9:00 AM) and is characterized by a broad spread of pickups across residential areas in the outer boroughs and Upper Manhattan. 
P.d.~B (top-right) is based on pickups during the evening commuter rush of the weekdays (5:00 PM – 6:00 PM). In stark contrast to the morning commute, the weekday evening rush features intensely concentrated demand in Midtown and Lower Manhattan as workers request rides from corporate hubs to return home. Noticeable hotspots also emerge at JFK and LGA airports. P.d.~C (bottom-left) is based on the weekend morning traffic (Weekend 8:00 AM - 9:00 AM). The distribution is noticeably dispersed across the boroughs, lacking the pronounced residential or corporate focal points of the weekdays. Because overall local demand is much lower and decentralized, JFK Airport stands out as the most prominent hotspot. P.d.~D (bottom-right) is based on Saturday evening (5:00 PM - 9:00 PM) demand, showing a renewed concentration of hotspots targeting dining, social, and entertainment districts in Manhattan and Brooklyn. Alongside this nightlife activity, JFK and LGA still maintain a strong, highly visible presence, capturing the influx of Saturday afternoon/evening flights. We remark that even though Newark airport corresponds to taxizone 1, it has probability $0$ in all four distributions. This is due to the fact that the pickup locations in the datasets we used are within the limits of the city of New York. 

\section{Analysis of algorithms VRRP and UC$k$M}\label{sec:vrrp+uckm}
We now analyze the algorithm VRRP that we used in our experiments. The main difference with the random placement algorithm is that it selects the $k$ points in the solution according to the variance-reduced probability distribution $\widetilde{P}_k$ of the original distribution $P_k$ over sets of $k$ points defined as follows. For each point $x\in \X$, the distribution $\widetilde{P}_k$ returns $\lfloor P(x)\cdot k\rfloor$ points. Each of the remaining $k-\sum_{x\in X}{\lfloor P(x)\cdot k}\rfloor$ points in the solution (if any) is selected independently and proportionally to the residual probabilities $P(x)-\lfloor P(x)\cdot k\rfloor /k$. I.e., point $x$ is selected with probability $\frac{P(x)-\lfloor P(x)\cdot k\rfloor /k}{\sum_{y\in \X}{\left(P(y)-\lfloor P(y)\cdot k\rfloor /k\right)}}$.

For a set $X$ of $k$ points, define the point-mass distribution $D_X$ that assigns value $D_X(x)$ to location $x\in \X$ equal to its multiplicity in $X$ over $k$. Then, the minimum-cost matching distance between two sets $X$ and $Y$ of $k$ points each is equal to $k$ times the Wasserstein distance of their point-mass distributions, i.e., $d_k(X,Y)=k\cdot d_W(D_X,D_Y)$. It will be convenient to work with the Wasserstein distance since we can use the triangle inequality and relate distances between sets of points to distances between (empirical) distributions and to distances between a set of points and a distribution.

\begin{theorem}
Consider a $k$-RP instance and let $O$ be the optimal $k$-multiset and $S$ the random $k$-multiset returned by algorithm VRRP. It holds that $\E_{S\sim \widetilde{P}_k}\E_{X\sim P_k}\left[d_k(S,X)\right] \leq 2\cdot \E_{X\sim P_k}\left[d_k(O,X)\right]$.
\end{theorem}

\begin{proof}
We will show that 
\begin{align}\label{eq:vrrp-vs-rp-inequality}
    \E_{S\sim \widetilde{P}_k}[d_k(S,O)] &\leq \E_{S\sim P_k}[d_k(S,O)].
\end{align}
Then, the desired inequality follows easily:
\begin{align*}
    \E_{S\sim \widetilde{P}_k}\E_{X\sim P_k}[d_k(S,X)] &\leq     \E_{S\sim \widetilde{P}_k}\E_{X\sim P_k}[d_k(S,O)+d_k(O,X)]\\
    &=\E_{S\sim \widetilde{P}_k}\E_{X\sim P_k}[d_k(S,O)]+\E_{S\sim \widetilde{P}_k}\E_{X\sim P_k}[d_k(O,X)]\\
    &=\E_{S\sim \widetilde{P}_k}[d_k(S,O)]+\E_{X\sim P_k}d_k(O,X]\\
    &\leq \E_{S\sim P_k}[d_k(S,O)]+\E_{X\sim P_k}d_k(O,X)]\\
    &= 2\cdot \E_{X\sim P_k}d_k(O,X)].
\end{align*}
The first inequality follows from triangle inequality for distance $d_k$ (from Lemma~\ref{lem:metric}) and the second one from inequality (\ref{eq:vrrp-vs-rp-inequality}). The first equality is due to linearity of expectation and the other two are obvious.

The rest of the proof is devoted to proving inequality (\ref{eq:vrrp-vs-rp-inequality}). We will use the following key statement in stochastic majorization.

\begin{theorem}[stochastic majorization, see~\cite{shaked2007stochastic}]\label{thm:stochastic-majorization}
Let $Y_1$, $Y_2$, ..., $Y_k$ be independent random vectors in $\R^n$, and $X_1$, $X_2$, ..., $X_k$ i.i.d.~random vectors whose common distribution is the average of the distributions of the $Y_j$'s, i.e., $\Pr[X_i=v] = \frac{1}{k}\sum_{j=1}^k{\Pr[Y_j=v]} $. Then, for every convex function $f:\R^n\rightarrow \R$, it holds 
\begin{align*}
    \E\left[f\left(\sum_{j=1}^k{Y_j}\right)\right] &\leq  \E\left[f\left(\sum_{j=1}^k{X_j}\right)\right].
\end{align*}
\end{theorem}

Let $N_S \in \R^n$ denote the random vector representing the integer counts of the $k$ points in a multiset $S$, such that its $x$-th component is the multiplicity of $x\in \X$ in $S$. The corresponding point-mass distribution is $D_S = N_S/k$.

By the definition of the Wasserstein distance as the value of a linear program (e.g., see~\cite{peyre2019computational}), we have that the function $g:\R^n\rightarrow \R$ defined as $g(D)=d_W(D,D_O)$ is convex. We define $f:\R^n\rightarrow \R$ as
\begin{align*}
    f(N) &= k \cdot g\left(N/k\right) = k \cdot d_W\left(N/k, D_O\right).
\end{align*}
Scaling the domain and multiplying by a positive constant preserves convexity. Hence, $f$ is also a convex function. By the relationship between the Wasserstein distance and the discrete matching distance, we have $f(N_S) = d_k(S, O)$.

To apply the stochastic majorization theorem, we decompose the count vectors $N_S$ for both distributions into sums of $k$ standard basis vectors.

For the standard random placement $S \sim P_k$, the multiset is formed by $k$ independent draws from the distribution $P$. We write $N_S = \sum_{j=1}^k X_j$, where each $X_j \in \R^n$ is an independent and identically distributed standard basis vector representing a single draw. Their common distribution is
\begin{align*}
    \Pr[X_j = e_x] &= P(x) \quad \text{for all } x \in X.
\end{align*}

For the variance-reduced placement $S \sim \widetilde{P}_k$, the algorithm deterministically places $c_x = \lfloor P(x) \cdot k \rfloor$ points at each location $x$, and draws the remaining $k' = k - \sum_{x \in X} c_x$ points independently from the residual probability $P'(x) = \frac{k P(x) - c_x}{k'}$. We decompose its count vector as $N_S = \sum_{j=1}^k Y_j$, where the $Y_j$'s are independent standard basis vectors constructed as follows:
\begin{itemize}
    \item For each $x \in X$, exactly $c_x$ of the vectors $Y_j$ are fixed deterministically to $e_x$ (i.e., $\Pr[Y_j = e_x] = 1$). This accounts for exactly $k - k'$ vectors.
    \item The remaining $k'$ vectors (if any) are i.i.d.~draws from the residual distribution $P'$, meaning for these vectors, $\Pr[Y_j = e_x] = P'(x) = \frac{k P(x) - c_x}{k'}$.
\end{itemize}

We now evaluate the average marginal distribution of the independent $Y_j$ sequence. For any standard basis vector $e_x$, the sum of the probabilities across all $k$ vectors is the deterministic count $c_x$ plus the expected residual count, i.e.,
\begin{align*}
    \sum_{j=1}^k \Pr[Y_j = e_x] &= c_x \cdot 1 + k' \cdot P'(x) = c_x + k' \left( \frac{k P(x) - c_x}{k'} \right) = k P(x).
\end{align*}
Dividing by $k$ yields the exact distribution of the $X_j$ vectors:
\begin{align*}
    \frac{1}{k} \sum_{j=1}^k \Pr[Y_j = e_x] &= P(x) = \Pr[X_1 = e_x].
\end{align*}

Because the $X_j$ vectors are i.i.d.~and their common distribution is exactly the average of the distributions of the independent $Y_j$ vectors, the conditions of Theorem~\ref{thm:stochastic-majorization} are satisfied. Applying the theorem to the convex function $f$ yields
\begin{align*}
    \E\left[f\left(\sum_{j=1}^k Y_j\right)\right] &\leq \E\left[f\left(\sum_{j=1}^k X_j\right)\right].
\end{align*}
Substituting the combined count vector $N_S$ back into the inequality gives
\begin{align*}
    \E_{S \sim \widetilde{P}_k}[f(N_S)] &\leq \E_{S \sim P_k}[f(N_S)].
\end{align*}
Since $f(N_S) = d_k(S, O)$, we obtain the desired inequality.
\end{proof}

\begin{corollary}\label{cor:3-apx-det-VRRP}
    Consider a $k$-RP instance and let $O$ be the optimal $k$-multiset. There exists a $k$-multiset $S$ consisting of points that are drawn with positive probability from the variance-reduced probability distribution $\widetilde{P}_k$ so that $\E_{X\sim P_k}\left[d_k(S,X)\right] \leq 2\cdot \E_{X\sim P_k}\left[d_k(O,X)\right]$.
\end{corollary}

We are now ready to analyze algorithm UC$k$M. We assume that the algorithm computes a $\rho$-approximate $k$-multiset under the objective $\min_{T\in \X_k}{d_W(D_T,P)}$. Our result for its performance for $k$-RP is as follows.

\begin{theorem}\label{thm:uckm}
Consider a $k$-RP instance and let $O$ be the optimal $k$-multiset and $S$ the $k$-multiset returned by algorithm UC$k$M. It holds that $\E_{X\sim P_k}\left[d_k(S,X)\right] \leq (\rho+2)\cdot \E_{X\sim P_k}\left[d_k(O,X)\right]$.
\end{theorem}

\begin{proof}
By the definition of the algorithm and its approximation guarantee, it holds that
\begin{align}\label{eq:uckm}
d_W(D_S,P) &\leq \rho\cdot \min_{T\in \X_k}{d_W(D_T,P)} \leq \rho\cdot d_W(D_O,P).
\end{align}
Thus, we have
\begin{align*}
    \E_{X\sim P_k}\left[d_k(S,X)\right] &= k\cdot \E_{X\sim P_k}[d_W(D_S,D_X)]\\
    &\leq k\cdot \E_{X\sim P_k}[d_W(D_S,P)+d_W(P,D_O)+d_W(D_O,D_X)]\\
    &\leq k\cdot \left(\rho\cdot d_W(D_O,P)+d_W(P,D_O)+\E_{X\sim P_k}[d_W(D_O,D_X)]\right)\\
    &\leq k\cdot \left(\rho+2\right)\cdot \E_{X\sim P_k}[d_W(D_O,D_X)]\\
    &= (\rho+2)\cdot \E_{X\sim P_k}[d_k(O,X)],
\end{align*}
as desired. The first and last equalities follow by the relation between the distances $d_W$ and $d_k$. The first inequality follows by the triangle inequality for the Wasserstein distance. The second one follows by Equation~(\ref{eq:uckm}), and the third one by Jensen's inequality applied to the convex function $d_W(D_O,\cdot)$, which implies that $d_W(D_O,P)\leq \E_{X\sim P_k}[d_W(D_O,D_X)]$.
\end{proof}

\section{A restricted variant of robotaxi placement}\label{sec:rrp}
We now consider the following restricted variant of robotaxi placement. An instance of $k$-RRP (standing for restricted robotaxi placement) consists of a metric space $(X,d)$ and a probability distribution over the points of $X$, together with a set of allowable points $A\subseteq X$. Our objective is to select a $k$-multiset $S$ with points from $A$ so that the quantity $\E_{X\sim P_k}[d_k(S,X)]$ is minimized. Essentially, this variant restricts the solution to consist of points in the allowable set $A$.

We define the following algorithm, which we refer to as {\em random restricted placement} (RRP). RRP first selects a random $k$-multiset $S$ according to the distribution $P_k$ (i.e., similarly to algorithm RP), and returns a $k$-multiset $T(S)$ of points from $A$ that minimizes $d_k(S,T(S))$.

\begin{theorem}
    Consider a $k$-RRP instance and let $O$ be the optimal $k$-multiset, $S$ the intermediate random $k$-multiset according to $P_k$ computed by algorithm RRP, and $T(S)$ the final $k$-multiset of points from $A$ it returns. It holds that $\E_{S\sim P_k}\E_{X\sim P_k}\left[d_k(T(S),X)\right] \leq 3\cdot \E_{X\sim P_k}\left[d_k(O,X)\right]$.
\end{theorem}

\begin{proof}
We have that 
\begin{align*}
    \E_{S\sim P_k}\E_{X\sim P_k}\left[d_k(T(S),X)\right] &\leq \E_{S\sim P_k}\E_{X\sim P_k}\left[d_k(T(S),S)+d_k(S,O)+d_k(O,X)\right]\\
    &=\E_{S\sim P_k}\left[d_k(T(S),S)\right]+2\cdot \E_{X\sim P_k}\left[d_k(O,X)\right]\\
    &\leq \E_{S\sim P_k}\E_{X\sim P_k}\left[d_k(O,S)\right]+2\cdot \E_{X\sim P_k}\left[d_k(O,X)\right]\\
    &=3\cdot \E_{X\sim P_k}\left[d_k(O,X)\right],
\end{align*}
as desired. The first inequality follows from the triangle inequality (Lemma~\ref{lem:metric}) and the second one from the definition of $k$-multiset $T(S)$, which implies that $d_k(T(S),S)\leq d_k(O,S)$. The equalities are straightforward, since both $k$-multisets $X$ and $S$ follow the same distribution $P_k$.
\end{proof}
\end{document}